\documentclass[runningheads]{llncs}
\usepackage{amsmath,amssymb}
\usepackage{bm}
\usepackage{quantikz}
\usepackage{braket}
\usepackage{xcolor}
\usepackage{booktabs}

\renewenvironment{example}
  {\vspace{0.25cm} \par\noindent\textbf{Example.}\itshape}
  {\par}

\newcommand{\CNOT}{\mathrm{CNOT}}
\newcommand{\calO}{\mathcal{O}}

\title{Constrained Quantum Optimization \\ meets Model Reduction}

\author{Max Tschaikowski\inst{1} \and Andrea Vandin\inst{2,3}}

\institute{
Sapienza University of Rome, Italy
\and
Scuola Superiore Sant'Anna, Italy 
\and
Technical University of Denmark
}

\begin{document}

\maketitle

\begin{abstract}
Quantum optimization algorithms promise advantages for difficult problems but are costly to simulate and analyze on classical machines. Recently, constrained quantum optimization has been investigated through the lens of Quantum Zeno dynamics, an approach which constrains the search to a subspace by means of quantum measurements. Exploiting that quantum measurements are projections, we propose a model reduction approach and show that simulations can be conducted in a lower-dimensional space. As possible applications, we demonstrate exponential state-space reduction of constrained quantum optimization in case of random 3-SAT and an agent coordination problem over graphs.
\end{abstract}

\section{Introduction}

Formal verification and analysis of collaborative adaptive systems require models that make safety and coordination constraints explicit while remaining amenable to scalable analysis. A central technique in theoretical computer science is to construct reduced models that preserve behavior relevant to the properties of interest, for instance via simulation or bisimulation-based quotients~\cite{Milner1989,BaierKatoen2008}. In this paper, we show that a closely related perspective applies to constrained quantum optimization: feasibility constraints induce an exact reduced quantum model that can be simulated more efficiently.

Quantum computing has gained remarkable momentum over the last decade and holds the promise of addressing computationally difficult problems such as combinatorial optimization. Among the most prominent approaches is the Quantum Approximate Optimization Algorithm (QAOA)~\cite{farhi2014quantum,Herman2023}, which encodes a classical optimization problem into a problem Hamiltonian and alternates it with a mixer Hamiltonian to steer the system toward high-quality solutions. In many applications, however, one is interested not only in optimal solutions but in feasible ones that satisfy additional constraints. This is particularly natural in collaborative systems, where feasibility encodes coordination requirements among agents. We follow this paradigm by studying constrained quantum optimization in case of a agent-coordination problem over graphs, viewed as a formal model of a multi-agent system.

In quantum computing, feasibility can be enforced either at the level of circuit design, for instance via constraint-preserving mixers~\cite{Fuchs2022}, or dynamically via quantum Zeno dynamics, where repeated projections confine the evolution to a safe subspace~\cite{FacchiPascazio2002,FacchiReview2007,Herman2023}. While the latter provides a flexible and declarative mechanism for enforcing constraints, classical simulation of quantum optimization remains costly because the evolution is represented in a space whose dimension grows exponentially with the number of qubits.

Our main observation is that the projector defining the Zeno subspace also induces an exact reduced model. Rather than representing the constrained dynamics in the full space, one can equivalently describe the evolution by working in a lower-dimensional space. This yields a compressed representation of the dynamics that captures exactly the behavior of the constrained system. As a result, quantum Zeno dynamics can be understood as a principled model-reduction technique for constrained quantum systems, enabling analysis and simulation to be carried out on a significantly smaller state space.

This reduction admits a natural interpretation in terms of behavioral equivalence. In concurrency theory, bisimulation identifies states that are indistinguishable with respect to observable behavior and yields quotient systems that preserve properties of interest~\cite{Milner1989,BaierKatoen2008}. While quantum dynamics is continuous rather than transition-based, our compression plays an analogous role: it removes behavior outside the safe subspace as irrelevant for the constrained evolution and yields an exact quotient of the dynamics within the safe region. In this sense, our construction can be viewed as a linear-algebraic analogue of bisimulation-based model reduction.

We provide (i) a formal reduction result showing that (Zeno-)constrained dynamics admits an exact compressed representation, and (ii) an application demonstrating exponential state-space reductions in case of random 3-SAT and a agent-coordination problem over random graphs. This places constrained quantum optimization within a broader verification-oriented perspective, where feasibility constraints induce invariant subspaces and exact reduced models, closely paralleling abstraction techniques in the analysis of collaborative adaptive systems.

\emph{Related work.} Results on quantum Zeno dynamics are well established in mathematical physics, where frequent projections lead to an effective evolution~\cite{FacchiPascazio2002,FacchiReview2007}, and product-formula constructions make this connection precise~\cite{Ichinose2006,BurgarthFacchi2020}. Recent work has explored Zeno dynamics for constrained quantum optimization~\cite{Herman2023}, alongside circuit-level approaches based on specialized mixers~\cite{Fuchs2022}. Our contribution is complementary: we make explicit the induced reduced dynamics and exploit it for classical simulation. In that respect, the present work relates to model reduction approaches in performance modeling~\cite{DBLP:journals/tcs/TschaikowskiT14,DBLP:journals/pe/TschaikowskiT17,DBLP:conf/popl/CardelliTTV16,DBLP:journals/tcad/0001MTA16}, computational biology~\cite{DBLP:conf/gecco/TognazziTTV17,DBLP:journals/bioinformatics/CardelliPTTVW21,Shizuka22} and engineering~\cite{DBLP:journals/tac/PappasS02,antoulas,DBLP:conf/cdc/TognazziTTV18,DBLP:journals/tac/CardelliGLTTV23}. Likewise, it closely aligns to quantum circuit reduction~\cite{qTACAS,DBLP:journals/tqc/BurgholzerJLTTW25}. However, while conceptually related, the foregoing approaches are different because they exploit invariant subspaces of the model, while we do not make such an assumption. Instead, we assume that the modeler is only interested in a certain subspace and thus is willing to lose information that lies outside of it, for instance because states outside the subspace are considered to be unsafe.

\paragraph{Synopsis.} Section~2 recalls the quantum-computing background needed in the paper, including quantum circuits, Schr\"odinger dynamics, and the formulation of QAOA for SAT. Section~3 introduces quantum Zeno reduction and proves that projection onto a feasible subspace induces an exact reduced dynamics that can be simulated in lower dimension. 
Section~4 evaluates the approach on constrained 3-SAT instances and on a agent-coordination benchmark encoded over graphs, demonstrating in some cases exponential state-space reductions. Section~5 concludes and discusses directions for future work.

\section{Background}\label{sec:quantum-systems}

This section provides a brief introduction to quantum computation that is needed to describe quantum optimization. To allow for a gentle discussion, we first present digital circuits and move afterwards to quantum circuits, where classical bit strings are replaced by unit vectors and classical gates by quantum gates.

\subsection{Digital Circuits}

Using the language of linear algebra, a classical bit is represented by a basis vector
\[
|0\rangle := \begin{pmatrix}1\\0\end{pmatrix},
\qquad
|1\rangle := \begin{pmatrix}0\\1\end{pmatrix}.
\]
For instance, the NOT gate can be represented as a matrix and visually as a circuit:
\[
X = \begin{pmatrix}0 & 1\\ 1 & 0\end{pmatrix},
\qquad
X|0\rangle=|1\rangle,\;\;
X|1\rangle=|0\rangle,
\qquad
\begin{quantikz}
\lstick{$|x\rangle$} & \gate{X} & \rstick{$|1-x\rangle$} \qw
\end{quantikz}
\]
Sequential composition corresponds to matrix multiplication and horizontal circuit composition. For instance,
\[
XX=I ,
\qquad 
\begin{quantikz}
\lstick{$|x\rangle$} & \gate{X} & \gate{X} & \rstick{$|x\rangle$} \qw
\end{quantikz}
\]
Instead, parallel composition corresponds to the tensor product. For two bits,
\[
|x_1x_0\rangle = |x_1\rangle  |x_0\rangle = |x_1\rangle \otimes |x_0\rangle.
\]
If $U$ and $V$ are single-bit gates, then their parallel composition satisfies
\[
(U\otimes V)(|x\rangle|y\rangle)
=
(U|x\rangle)\otimes(V|y\rangle).
\]
For example, $(X\otimes X)|x_1x_0\rangle = (X|x_1\rangle)\otimes(X|x_0\rangle) = |1-x_1\rangle  |1-x_0\rangle $ and can be depicted as
\[
\begin{quantikz}
\lstick{$|x_1\rangle$} & \gate{X} & \qw \\
\lstick{$|x_0\rangle$} & \gate{X} & \qw
\end{quantikz}
\]

To create more complex computations, one needs to use multi-bit gates. A famous example is the controlled NOT gate (CNOT):
\[
\CNOT\bigl(\ket{x_1}\ket{x_0}\bigr)
=
\ket{x_1}\ket{x_1\oplus x_0},
\]
where \(\oplus\) denotes XOR. The first bit is copied as a control, and the second bit is flipped exactly when the first one is \(1\). 

\begin{example}
Parallel composition of a CNOT with an \(X\) gate on a third bit, followed by an additional \(X\) gate on the second (target) bit, can be depicted as
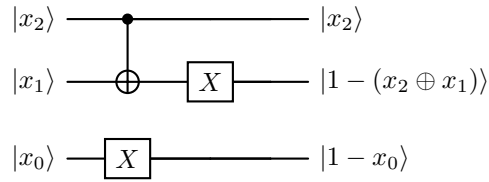
\begin{center}
\begin{quantikz}
\lstick{$\ket{x_2}$} & \ctrl{1} & \qw      & \qw      & \rstick{$\ket{x_2}$}\\
\lstick{$\ket{x_1}$} & \targ{}  & \gate{X} & \qw      & \rstick{$\ket{1-(x_2\oplus x_1)}$}\\
\lstick{$\ket{x_0}$} & \gate{X} & \qw      & \qw      & \rstick{$\ket{1-x_0}$}
\end{quantikz}
\end{center}
Algebraically, the circuit can be written as
\[
\bigl((I \otimes X \otimes I)\circ(\CNOT \otimes X)\bigr)
\bigl(\ket{x_2}\ket{x_1}\ket{x_0}\bigr)
=
\ket{x_2}\ket{1-(x_2\oplus x_1)}\ket{1-x_0} ,
\]
where $\circ$ emphasizes sequential composition and is often suppressed because it corresponds to matrix multiplication. Notice that the first layer $\CNOT \otimes X$ is applied first and appears therefore in the above equation to the right of the second layer $I \otimes X \otimes I$, which is applied to the result of the first layer.
\end{example}


\subsection{Quantum Circuits}

Unlike a classic computer whose state is given by a bit string $\ket{x}$, a quantum computer can attain several bit strings simultaneously. The linear algebra perspective is helpful to formalize what it means to attain several bit strings simultaneously. Indeed, in the case of digital circuits, bits $\ket{0}$ and $\ket{1}$ were expressed by vectors $\binom{1}{0}$ and $\binom{0}{1}$, respectively. Intuitively, it is therefore natural to interpret the vector 
\begin{align*}
    \binom{1}{1} & = \binom{1}{0} + \binom{0}{1} = \ket{0} + \ket{1}
\end{align*}
as the quantum state that attains bits $\ket{0}$ and $\ket{1}$ at the same time. By the laws of quantum mechanics, the state of a qubit $\ket{\psi}$ is actually a sum with complex-valued coefficients, i.e., $\ket{\psi} = \alpha_0 \ket{0} + \alpha_1 \ket{1}$. When measured, a qubit attains the value $\ket{x}$ with probability $|\alpha_x|^2$. The constraint of $|\alpha_x|^2$ being a probability implies that $\sum_x |\alpha_x|^2 = 1$.

\begin{example}
A central gate in quantum computing realizing quantum parallelism is the Hadamard gate
\[
H=
\frac{1}{\sqrt{2}}
\begin{pmatrix}
1 & 1\\
1 & -1
\end{pmatrix}
\]
Notice that $H\ket{0} = \frac{1}{\sqrt{2}}(|0\rangle+|1\rangle)$ and that $(\frac{1}{\sqrt{2}})^2 = \frac{1}{2}$. That is, measuring the quantum state $H \ket{0}$ returns $\ket{0}$ or $\ket{1}$ with probability $\frac{1}{2}$.
\end{example}

The above discussion yields the following.

\begin{definition}
A quantum state on $n$ qubits is a unit vector in $\mathbb{C}^{2^n}$. It has the form
\[
|\psi\rangle
=
\sum_{x\in\{0,1\}^n}
\alpha_x |x\rangle,
\qquad
\sum_x |\alpha_x|^2 = 1.
\]
The coefficients $\alpha_x \in \mathbb{C}$ are called amplitudes. Measuring $|\psi\rangle$ in the computational basis yields outcome $x$ with probability $|\alpha_x|^2$. 
\end{definition}

A quantum circuit must map quantum states to quantum states. Since states are unit vectors, admissible transformations must preserve the Euclidean norm. The linear maps with this property are precisely the unitary matrices.

\begin{definition}
A quantum circuit on $n$ qubits is a unitary matrix
\[
U \in \mathbb{C}^{2^n\times 2^n},
\qquad
U^\dagger U = I ,
\]
where $\dagger$ denotes matrix transposition with complex conjugation. Sequential and parallel composition are defined, respectively, via matrix multiplication and the tensor product $\otimes$. 
\end{definition}

\begin{example}
Note that that quantum gates and circuits do not only output quantum states, they also take these as inputs. The Hadamard gate could be for instance applied to the quantum states 
\[
\ket{+} = \tfrac{1}{\sqrt{2}} \ket{0} + \tfrac{1}{\sqrt{2}}\ket{1} \quad \text{and} \quad 
\ket{-} = \tfrac{1}{\sqrt{2}} \ket{0} - \tfrac{1}{\sqrt{2}}\ket{1} .
\]
In case of $\ket{+}$, e.g., we obtain 
\begin{align*}
H\!\big(\tfrac{1}{\sqrt{2}} \ket{0}+ \tfrac{1}{\sqrt{2}}\ket{1} \big) 
&=
\tfrac{1}{\sqrt{2}} H\ket{0} + \tfrac{1}{\sqrt{2}} H\ket{1}
=
\tfrac12 (\ket{0}+\ket{1})
+
\tfrac12 (\ket{0}-\ket{1})
=
\ket{0}.
\end{align*}
Applying two Hadamard gates in parallel, instead, gives
\begin{align*}
H^{\otimes 2}|0\rangle^{\otimes 2} & = (H \otimes H) (|0\rangle \otimes |0\rangle) \\
&= (H|0\rangle)\otimes(H|0\rangle) \\
&= \Big(\frac{1}{\sqrt{2}}(|0\rangle+|1\rangle)\Big)
   \otimes
   \Big(\frac{1}{\sqrt{2}}(|0\rangle+|1\rangle)\Big) \\
&= \frac{1}{2}
   \big(
   |0\rangle \otimes |0\rangle + |0\rangle \otimes |1\rangle + |1\rangle \otimes |0\rangle + |1\rangle \otimes |1\rangle
   \big) \\
&= \frac{1}{2}
   \big(
   |00\rangle + |01\rangle + |10\rangle + |11\rangle
   \big).
\end{align*}
More generally, $H^{\otimes n}|0\rangle^{\otimes n} = \frac{1}{\sqrt{2^n}} \sum_{x\in\{0,1\}^n} |x\rangle$. That is, Hadamard gates can be used to realize uniform superpositions of all quantum base states. 
\end{example}

\subsection{Schr\"odinger's Equation}

From a physical perspective, it is not the quantum circuit that is directly implemented in quantum hardware. Indeed, the results of quantum circuits are solutions of Schr\"odinger's equation 
\begin{align}\label{eq:schroedinger}
\partial_t |\psi_t \rangle = -i H_U |\psi_t \rangle 
\end{align}
where $H_U$ is a Hermitian matrix. That is, for each quantum gate $U$, there exists a Hermitian matrix $H_U$ (the so-called  Hamiltonian) and an execution time $s$ such that $U \ket{\psi_0} = \ket{\psi_s}$, with $\ket{\psi_t}$ satisfying~(\ref{eq:schroedinger}) for all $0 \leq t \leq s$, and $\ket{\psi_0}$ being the input to $U$. More generally, quantum circuits are described by time-dependent control fields (e.g., microwave pulses or laser interactions) that realize a time-dependent Hamiltonian $H_U(t)$, where time-dependency accounts for different layers of the quantum circuit $U$.

Being linear, we recall that the solution of~(\ref{eq:schroedinger}) at $t$ is given by $\exp(-i t H_U)$, where the matrix exponential is defined as
\[
\exp(M) = \sum_{k = 0}^\infty \frac{1}{k!} M^k
\]

\begin{example}
Assume we wish to realize the Hadamard gate $H$. To this end, we consider the Hamiltonian $H_U = (I - H) / 2$, where $I$ is the identity. Then, picking as execution time $s = \pi / 2$, it can be shown that $H \ket{\psi_0} = \ket{\psi_s}$, where $\ket{\psi_t}$ obeys the Schrödinger equation for Hamiltonian $H_U$. Specifically, since $I$ commutes with $H$, we may split the exponential:
\begin{align*}
\exp({-i s H_U})
&= \exp\Big({-i\frac{\pi}{2}(I-H)}\Big)
   = \exp\Big({-i\frac{\pi}{2}I}\Big) \exp\Big({i\frac{\pi}{2}H}\Big).
\end{align*}
Using $H^2 = I$, we infer that
\begin{align*}
\exp\Big({i\frac{\pi}{2}H}\Big)
&= \cos\!\left(\frac{\pi}{2}\right) I
   + i\sin\!\left(\frac{\pi}{2}\right) H
   = iH.
\end{align*}
With $\exp({-i\frac{\pi}{2}I}) = -i\,I$, we obtain overall $\exp(-i \frac{\pi}{2} H_U) = (-iI)(iH) = H$.
\end{example}

\subsection{Quantum Optimization}

Quantum approximate optimization algorithm (QAOA) is conviniently described by Schrödinger's equation where the Hamiltonian $H$ varies in time, alternating between a problem Hamiltonian $H_P$ and a begin Hamiltonian $H_B$. 


\begin{definition}[QAOA~\cite{farhi2014quantum}] \label{def:qaoa}
For a problem Hamiltonian $H_P$ and a begin Hamiltonian $H_B$, fix the unitary matrices
\begin{align*}
U_B(\delta) & = \exp(-i \delta H_B) & \text{and} &  &
U_P(\delta) & = \exp(-i \delta H_P)
\end{align*}
where $\delta > 0$ is a sufficiently small time step and $\exp(A)$ is the matrix exponential. For a  sequence of natural numbers $(k_i,l_i)_{i=1}^\kappa$ of length $\kappa \geq 1$, we define 
\begin{equation}\label{eq:qaoa:expanded}
        \ket{w_\kappa} = U_B(\delta)^{k_\kappa} U_P(\delta)^{l_\kappa} \cdot \ldots \cdot U_B(\delta)^{k_1} U_P(\delta)^{l_1} \ket{\psi}
\end{equation}    
The QAOA with $\kappa \geq 1$ stages is then given by $\max \{ \bra{w_\kappa} H_P \ket{w_\kappa} \mid (k_i,l_i)_{i=1}^\kappa \}$. 
\end{definition}

As anticipated,~(\ref{eq:qaoa:expanded}) corresponds to the evaluation of Schrödinger's equation in case where the Hamiltonian $H$ is alternating between $H_P$ and $H_B$. Indeed, $U_B(\delta) \ket{\psi_0} = \ket{\psi_\delta}$, where $\ket{\psi_t}$ solves $\partial_t \ket{\psi_t} = -i H_B \ket{\psi_t}$ on interval $[0;\delta]$. A product $U_P(\delta') U_B(\delta) \ket{\psi_0}$, instead, is obtained by solving $\partial_t \ket{\psi_t} = -i H_P \ket{\psi_t}$ on the interval $[\delta ; \delta + \delta']$.

While the problem Hamiltonian $H_P$ depends on the optimization problem at hand, the choice of the begin Hamiltonian $H_B$ is informed by the so-called adiabatic theorem, a result that identifies conditions under which QAOA returns a global optimum. A common heuristic is to pick $H_B$ such that $H_B$ and $H_P$ do not diagonalize over a common basis~\cite{farhi2000quantum,farhi2014quantum} and to assume without loss of generality that $\ket{\psi_0} =  \sum_x \ket{x} / \sqrt{N}$ is the unique maximal eigenvector of $H_B$.

We next discuss QAOA in the case where it is applied to solve SAT, an NP-complete problem~\cite{sipser1996introduction}. 

\begin{definition}[Problem Hamiltonian SAT]\label{def:ham:sat}
For a boolean formula $\phi = \bigwedge_{i = 1}^M C_i$, where $C_i$ is a clause over $n$ boolean variables, the problem Hamiltonian is given by $H_P = \sum_i H_i$, where
        \[
        H_i \ket{x} =
        \begin{cases}
        \ket{x} & , \ C_i(x) \text{ is true} \\
        0 & , \ C_i(x) \text{ is false} \\
        \end{cases}
        \]
        for any $x \in \{0,1\}^n$ representing a boolean assignment.
\end{definition}

Following this definition, it can be shown that the QAOA $\bra{w_\kappa} H_P \ket{w_\kappa}$ from Definition~\ref{def:qaoa} corresponds to the average number of satisfied clauses. Under certain assumptions, QAOA finds a global optimum in case of a sufficiently large value of $\kappa$~\cite{farhi2000quantum,farhi2014quantum}.

\begin{example}
Consider the boolean formula $\Phi = (\neg x_0 \vee \neg x_1)\wedge(x_0\vee x_1)$. The first clause $C_1(x_1 x_0) = (\neg x_0 \vee \neg x_1)$ is true whenever $x_1 = 0$ or $x_0 = 0$; instead, the second clause is true whenever $x_1 = 1$ or $x_0 = 1$. Hence, $\Phi$ encodes $x_1 \texttt{ xor } x_2$, i.e., satisfying both clauses is only possible if $x_0 \neq x_1$. Overall, we obtain 
\[
H_P|00\rangle = 1|00\rangle, \ H_P|01\rangle = 2|01\rangle, \ H_P|10\rangle = 2|10\rangle, \ H_P|11\rangle = 1|11\rangle .
\]
Recalling that $|x \rangle = e_{x + 1}$, where $e_i$ is the $i$-th standard base vector, while $x$ is the binary representation of an integer, the above line translates to
\[
H_P e_1  = 1 e_1, \ H_P e_2 = 2 e_2, \ H_P e_3 = 2 e_3, \ H_P e_4 = 1 e_4 .
\]
Being diagonal, $H_P$ diagonalizes over the standard basis, i.e., the one induced by $|0\rangle$ and $|1\rangle$. Recalling that $H_P$ and $H_B$ should not diagonalize over the same basis, a common choice for $H_B$ is $X \otimes I + I \otimes X = X_1 + X_0$, giving us
\begin{align}\label{eq:HB:small}
H_P & = 
\begin{pmatrix}
1 & 0   & 0   & 0 \\
0   & 2 & 0   & 0 \\
0   & 0   & 2 & 0 \\
0   & 0   & 0   & 1
\end{pmatrix} 
\quad
\text{and}
\quad
H_B = 
\begin{pmatrix}
0 & 1 & 1 & 0 \\
1 & 0 & 0 & 1 \\
1 & 0 & 0 & 1 \\
0 & 1 & 1 & 0
\end{pmatrix}
\end{align}
We notice that $H_B$ does not diagonalize over the standard basis. Indeed, it can be shown that $H_B$ diagonalizes over the basis induced by $|+\rangle$ and $|-\rangle$. For instance, since $X |+\rangle = |+\rangle$ and $X |-\rangle = (-1) | - \rangle$, one readily infers that 
\[
H_B |- +\rangle = (X \otimes I + I \otimes X) |-\rangle \otimes |+\rangle = ((-1)|- \rangle \otimes |+\rangle) +  (|- \rangle \otimes |+\rangle) = 0
\]
\end{example}

\section{Quantum Zeno Reduction}

We next introduce quantum Zeno reduction. To facilitate the discussion, we shall use the SAT example from before.

\begin{example}
Recall our SAT example $\Phi = (\neg x_0 \vee \neg x_1) \land (x_0\vee x_1) = C_1 \land C_2$. In order to find a satisfying assignment of the entire formula $\Phi$, we restrict the search space to the boolean assignments satisfying the first clause that satisfy $C_1$. Put different, we consider as (safe) Zeno space
\begin{align*}
Z & = \mathrm{span} \{ | x \rangle \mid C_1(x) = 1 \} \\
& = \{ \alpha_{00} |00\rangle + \alpha_{01} |01\rangle + \alpha_{10} |10\rangle \mid \alpha_{00}, \alpha_{01}, \alpha_{10} \in \mathbb{C} \} \\
& = \{ \alpha_{1} e_1 + \alpha_{2} e_2 + \alpha_{e} e_3 \mid \alpha_{1}, \alpha_{2}, \alpha_{3} \in \mathbb{C} \} ,
\end{align*}
where $\mathrm{span}(A)$ denotes the span of a set of vectors, i.e., all finite linear combinations $\sum_i \alpha_i a_i$ where $a_i \in A$.   
\end{example}

Quantum Zeno projects Schrödinger's equation~(\ref{eq:schroedinger}) onto a subspace of the entire space, e.g., $Z$ in the case of our SAT example. 

\begin{definition}
Given a subspace of interest $Z$, quantum Zeno replaces the original Hamiltonian $H_U$ in Schrödinger's equation~(\ref{eq:schroedinger}) with the projected Hamiltonian $H_Z = P_Z H_U P_Z$, where $P_Z$ is the projection onto $Z$. Specifically, if $S$ is an orthonormal basis of subspace $Z$, then $P_Z = S S^\dagger$.
\end{definition}

\begin{example}
Consider matrix $S = (e_1, e_2, e_3)$ whose columns form an orthonormal basis of the safe space in the last example. Hence, the projection $P_Z$ of our SAT example is given by $P_Z = S S^\dagger = (e_1, e_2, e_3, 0)$, where $0$ denotes the zero column with four coordinates. As expected, $P_Z w$ sets the fourth coordinate in $w$ to zero because $e_4 = |11\rangle$ does not satisfy the first clause of $\Phi$. Recall that this was the original idea behind quantum Zeno: focus on a subspace that is of interest and remove everything else. 

Since $H$ alternates between $H_P$ and $H_B$ in the case of QAOA, we obtain two projected Hamiltonians:
\[
H_{P,Z} = P_Z H_P P_Z = 
\begin{pmatrix}
1 & 0   & 0   & 0 \\
0   & 2 & 0   & 0 \\
0   & 0   & 2 & 0 \\
0   & 0   & 0   & 0
\end{pmatrix} 
\quad
\text{and}
\quad
H_{B,Z} = P_Z H_B P_Z = 
\begin{pmatrix}
0 & 1 & 1 & 0 \\
1 & 0 & 0 & 0 \\
1 & 0 & 0 & 0 \\
0 & 0 & 0 & 0
\end{pmatrix}
\]
The constrained QAOA described by the Zeno subspace is thus given by
\begin{align*}
U_{B,Z}(\delta) & = \exp(-i \delta H_{ZB}) & \text{and} &  &
U_{P,Z}(\delta) & = \exp(-i \delta H_{ZP})
\end{align*}
\end{example}

In the above example, projecting onto a subspace led to zero columns and zero rows in Hamiltonians. This is because the calculations are confined to a subspace of smaller dimension. Rather than working in a subspace of a bigger state space, we transfer the entire calculation to a smaller state space. This makes it possible to simulate quantum algorithms more efficiently.

\begin{theorem}\label{thm:red:zeno}
Let $P_Z = S S^\dagger$ with $S^\dagger S = I$ describe the projection onto the safe subspace $Z$ and let $H_Z = P_Z H P_Z$ and $\widehat{H}_Z = S^\dagger H S$. Then, if $|\widehat{\psi}_0\rangle = S^\dagger |\psi_0\rangle$, the reduced and the original Zeno equation
\[
\partial_t |{\psi}_t \rangle = -i H_Z |\psi_t \rangle \quad \text{and} \quad
\partial_t |\widehat{\psi}_t \rangle = -i \widehat{H}_Z |\widehat{\psi}_t \rangle 
\]
satisfy $\ket{\psi_t} - r = S \ket{\widehat{\psi}_t}$, where $r = (I - P_Z) \ket{\psi_0}$. On the level of circuits, we obtain
\begin{align*}
U_Z(t) \ket{\psi_0} - r & = S \ket{\widehat{\psi}_t} = S \exp(-it \widehat{H}_Z) S^\dagger \ket{\psi_0} = S \hat{U}_{Z}(t) S^\dagger \ket{\psi_0} ,
\end{align*}
where $U_Z(t) = \exp(-it H_Z)$ and $\hat{U}_{Z}(t) = \exp(-it \widehat{H}_Z)$. Notice that $U_Z(t)$ is the original Zeno circuit, while $\hat{U}_Z(t)$ is the reduced Zeno circuit.
\end{theorem}

Intuitively, the remainder of $\ket{\psi_0}$ outside the Zeno space, $r = (I - P_Z) \ket{\psi_0}$, is not altered by the Zeno dynamics, while the part in the Zeno space, $P_Z \ket{\psi_0}$, gets propagated by $H_Z$ and equals $S |\widehat{\psi}_t \rangle$, i.e., can be obtained by solving the reduced Zeno equation described by $\widehat{H}_Z$. The case of QAOA that considers alternating Hamiltonians, with a common Zeno space, can be obtained by applying the above result iteratively and recalling $S^\dagger S = I$. Specifically, one obtains 
\begin{align}\label{eq:cor}
U_{k,Z}(t_k) \cdots U_{1,Z}(t_1) \ket{\psi_0} - r & =  S \Big( \hat{U}_{k,Z}(t_k) \cdots \hat{U}_{1,Z}(t_1) \Big) S^\dagger \ket{\psi_0} .
\end{align}
Intuitively, the above equation goes first from the original space to the reduced space via $S^\dagger$, then we perform a calculation using the reduced circuits $\hat{U}_{k,Z}(t_k) \cdots \hat{U}_{1,Z}(t_1)$, and lift at the end the final result from the reduced space to the original space via $S$.

\begin{example}
Consider the Zeno Hamiltonians from the last example 
\[
H_{P,Z} = P_Z H_P P_Z = 
\begin{pmatrix}
1 & 0   & 0   & 0 \\
0   & 2 & 0   & 0 \\
0   & 0   & 2 & 0 \\
0   & 0   & 0   & 0
\end{pmatrix} 
\quad
\text{and}
\quad
H_{B,Z} = P_Z H_B P_Z = 
\begin{pmatrix}
0 & 1 & 1 & 0 \\
1 & 0 & 0 & 0 \\
1 & 0 & 0 & 0 \\
0 & 0 & 0 & 0
\end{pmatrix}
\]
With $S = (e_1, e_2, e_3, 0)$ and $P_Z = SS^\dagger$, the respective reduced Zeno Hamiltonians are
\[
\hat{H}_{P,Z} = S^\dagger H_{P} S = 
\begin{pmatrix}
1 & 0   & 0   \\
0   & 2 & 0   \\
0   & 0   & 2 \\
\end{pmatrix} 
\quad
\text{and}
\quad
\hat{H}_{B,Z} = S^\dagger H_{B} S = 
\begin{pmatrix}
0 & 1 & 1 \\
1 & 0 & 0 \\
1 & 0 & 0 \\
\end{pmatrix} 
\]
That is, we reduce the dimension from $4$ to $3$. For $\ket{\psi_0} = (\ket{00} + \ket{01} + \ket{11}) / \sqrt{3}$, one then obtains
\[
\ket{\widehat{\psi}_0} = S^\dagger \big( \tfrac{1}{\sqrt{3}},\tfrac{1}{\sqrt{3}},0,\tfrac{1}{\sqrt{3}})^T = (\tfrac{1}{\sqrt{3}},\tfrac{1}{\sqrt{3}},0 \big)^T
\]
In the case of the diagonal matrix $\hat{H}_{B,Z}$, one can readily infer that the solution of $\partial_t \ket{\widehat{\psi}_t} = \hat{H}_{B,Z} \ket{\widehat{\psi}_t}$ is given by $\ket{\widehat{\psi}_t} = \big (\tfrac{1}{\sqrt{3}} e^{-it}, \tfrac{1}{\sqrt{3}} e^{-2it}, 0 \big)^T$.
\end{example}

\subsection{Efficient Simulation of Constrained Quantum Optimization}\label{sec:qaoa:red}

We next apply quantum Zeno reduction to speed-up the simulation of constrained quantum optimization on classical machines.

\emph{Simulation costs of quantum optimization.} Recalling Definition~\ref{def:qaoa}, we observe that a simulation of QAOA on a classical computer requires to determine the value of $\nu = \sum_{i=1}^\kappa (k_i + l_i)$ matrix-vector multiplications
\begin{equation}
        \ket{w_\kappa} = U_B(\delta)^{k_\kappa} U_P(\delta)^{l_\kappa} \cdot \ldots \cdot U_B(\delta)^{k_1} U_P(\delta)^{l_1} \ket{\psi_0}
\end{equation}    
Unitary matrices $U_B(\delta)$ and $U_P(\delta)$ are the solutions of the respective Schrödinger equations~(\ref{eq:schroedinger}) over time interval $[0;\delta]$, involving $2^n \times 2^n$ Hamiltonian matrices. Overall, a direct implementation comes with the following simulation costs:
\begin{itemize}
    \item The $\nu$ matrix-vector multiplications cost $\calO(\nu 2^{2 n})$.
    \item The computation of $U(\delta)$ costs $\calO(s 2^{2 n})$, where $s = \lceil \delta / \tau \rceil$ and $0 < \tau \leq \delta$ is, respectively, the number and the length of discretization steps used to approximate the solution of Schrödinger's equation via a numerical solver~\cite{BlumCuckerShubSmale1998}.
\end{itemize}

\emph{Simulation costs of constrained quantum optimization.} We next consider the case where QAOA is simulated in presence of constraints, allowing us to apply the quantum Zeno reduction result from Theorem~\ref{thm:red:zeno}. In this case, the $2^n \times 2^n$ matrices $U_P(\delta)$ and $U_B(\delta)$ are replaced, respectively, with the $d \times d$ matrices $U_{P,Z}(\delta)$ and $U_{B,Z}(\delta)$. Here, $d \leq 2^n$ is the dimension of the (safe) Zeno subspace $Z$, hopefully substantially smaller than $2^n$. 

The following result can be shown.

\begin{theorem}\label{thm:qaoa:red}
Using the notation from Theorem~\ref{thm:red:zeno} and~(\ref{eq:cor}), the following holds.    
\begin{enumerate}
    \item If $U_{P,Z}(\delta)$, $U_{B,Z}(\delta)$, $S$ and $\ket{\widehat{\psi}_0} = S^\dagger \ket{\psi_0}$  from~(\ref{eq:cor}) are known, the constrained quantum optimization computes in $\calO(\nu d^2)$, where $\nu = \sum_{i=1}^\kappa (k_i + l_i)$ is the number of matrix-vector multiplications in
    \begin{equation*}
            \ket{w_\kappa} = S \big( U_{B,Z}(\delta)^{k_\kappa} U_{P,Z}(\delta)^{l_\kappa} \cdot \ldots \cdot U_{B,Z}(\delta)^{k_1} U_{P,Z}(\delta)^{l_1} \big) S^\dagger \ket{\psi_0}
    \end{equation*}        
    
    \item Assuming that $H_Z$ is available, the computation of $U_{Z}(\delta)$ costs $\calO(s d^{2})$, where $s = \delta / \tau$ and $0 < \tau \leq \delta$ is, respectively, the number and the length of discretization steps used to approximate the solution of Schrödinger's equation.    

    \item Assuming that $S$ is available, the computation of the reduced Hamiltonian
    \[
    H_Z = S^\dagger H S
    \]
    can be done in $\calO(d 2^{2n})$. In the special case where $S$ consists of standard basis vectors, the cost is $\calO(d^2)$, provided that $H$ is available or that each of its entries can be computed independently from the others in constant time. 
\end{enumerate}
\end{theorem}

We next comment on the theorem. Statement 1. and 2. ensure that constrained quantum optimization can be obtained by working in the smaller Zeno space $Z$ of dimension $d \leq 2^n$. Both statement 1. and 2. work under the assumption that the reduced Hamiltonians $H_{P,Z}$ and $H_{B,Z}$ are available. Their computation is addressed in statement 3. which ensures that they can be computed in $\calO(d 2^{2n})$ steps. Statement 3. is a no free lunch result. Indeed, in order to bring constrained quantum optimization to the smaller space $Z$, we need to compute the reduced Hamiltonian $H_Z = S^\dagger H S$ which involves original Hamiltonian $H$ of size $2^n \times 2^n$. The respective computational cost of $\calO(d 2^{2n})$ allows for speed-ups when $\min\{\nu,s\} \gg d$. 

We notice that the above result ensures substantial speeds when $d \ll 2^n$ and $S$ consists of standard basis vectors only. Such a case will be discussed in Section~\ref{sec:eval}. We finish by noting that the availability of matrix $S$ corresponds to the availability of the (safe) Zeno space which, in turn, is assumed to be given in constrained optimization.

\section{Evaluation}\label{sec:eval}


In the following, we shall consider 3-SAT instances over $n$ variables and $m$ clauses, that is
\[
\Phi(x) = C_1(x) \land C_2(x) \land \ldots \land C_m(x) , \text{ with } x \in \{0,1\}^n
\]
For each $\Phi(x)$, we then consider as constraint the sub formula
\[
\Phi_k(x) = C_1(x) \land C_2(x) \land \ldots \land C_k(x) , \text{ with } x \in \{0,1\}^n \text{ and } k < m .
\]

All experiments reported below terminated within one second and were conducted on a 8-core machine with 2.7\,GHz and 16\,GB of RAM.

\setlength{\tabcolsep}{6pt}
\begin{table*}[t]
\centering
\begin{tabular}{@{}r r rrr rrr @{}}
\toprule
& & \multicolumn{3}{c}{\textbf{$m = 4.0n$}}
& \multicolumn{3}{c}{\textbf{$m = 5.0n$}} \\
\cmidrule(lr){3-5}\cmidrule(lr){6-8}
$n$ & $2^n$ & $k=0.7m$ & $k=0.8m$ & $k=0.9m$
& $k=0.7m$ & $k=0.8m$ & $k=0.9m$ \\
\midrule
10  & 1024  & 29.03  & 15.32  & 8.36    & 9.59    & 4.50  & 2.53 \\
11  & 2048  & 39.82  & 21.17  & 11.23   & 12.11   & 6.27  & 3.91 \\
12  & 4096  & 46.50  & 25.51  & 13.94   & 15.81   & 6.90  & 3.14 \\
13  & 8192  & 65.66  & 27.77  & 13.21   & 24.06   & 8.32  & 3.05 \\
14  & 16384 & 101.01 & 46.40  & 19.42   & 18.74   & 8.46  & 3.83 \\
15  & 32768 & 112.89 & 50.04  & 22.38   & 29.19   & 10.74  & 4.63 \\
\bottomrule
\end{tabular}
\caption{Random 3-SAT: Original search space dimension, $2^n$, and the average reduced search space dimension. Each average was computed across $100$ random 3-SAT instances $\Phi$ having $n$ boolean variables and $m$ clauses, with constraints $\Phi_k$ given by the first $k$ clauses of $\Phi$. 
}\label{tab:random:sat}
\end{table*}

\emph{{3-SAT benchmark.}} As first benchmark, we consider random 3-SAT formulas consisting of conjunctions of $m$ clauses which, in turn, consist of there randomly picked literals, i.e., boolean variables $x_i$ whose indices were picked uniformly from $0,\ldots,n-1$ and that were negated with probability of $\tfrac{1}{2}$. The respective results are reported in Table~\ref{tab:random:sat} and demonstrate that we can obtain substantial reductions of the search spaces. The reduced search spaces are given by the spans of satisfying boolean assignments of $\Phi_k$, that is 
\[
\mathrm{span} \{ \ket{x} \mid \Phi_k(x) = 1 \} .
\]
The satisfying boolean assignments were computed by a Python implementation of \#DPLL~\cite{BacchusDalmaoPitassi2003}.

\emph{{Agent cooperation benchmark.}} As second benchmark, we model coordination constraints among cooperating agents via a family of 3-SAT formulas derived from random graphs. For $n$, we generated $100$ Erd\H{o}s--R\'enyi graph with $n$ locations. We associated a Boolean variable $x_i$ at node $i \in V$, where $x_i = 1$ indicates that a agent is exploring at location $i$, while $x_i = 0$ means that the agent acts as passive relay. For every node $j$ and every pair of its neighbors $i,k \in N(j)$, we introduced then the clauses
\[
(\neg x_i \lor \neg x_j \lor \neg x_k) \;\land\; (x_i \lor x_j \lor x_k),
\]
which jointly enforce that any three neighboring agents $i,j,k$ are not all in the same state. Intuitively, this prevents local congestion (all exploring) as well as inactivity (all relaying) within two-hop neighborhoods. The resulting formula is an instance of the not-all-equal 3-SAT problem, where each clause requires that not all literals take the same truth value. It is well to be NP-complete~\cite{Schaefer1978}, and thus our construction yields a natural and computationally challenging benchmark for evaluating coordination under local interaction constraints. The results are reported in Table~\ref{tab:robots} and demonstrate, akin to random 3-SAT, that we can obtain substantial reductions of the search spaces.

\setlength{\tabcolsep}{6pt}
\begin{table*}[t]
\centering
\begin{tabular}{@{}r r rrr rrr @{}}
\toprule
& & \multicolumn{3}{c}{\textbf{$p = 0.3$}}
& \multicolumn{3}{c}{\textbf{$p = 0.4$}} \\
\cmidrule(lr){3-5}\cmidrule(lr){6-8}
$n$ & $2^n$ & $k=0.7m$ & $k=0.8m$ & $k=0.9m$
& $k=0.7m$ & $k=0.8m$ & $k=0.9m$ \\
\midrule
10  & 1024  & 103.41  & 86.09  & 47.44   & 57.05    & 25.86  & 4.77 \\
11  & 2048  & 144.29  & 72.85  & 29.37   & 41.84    & 16.22  & 3.97 \\
12  & 4096  & 157.86  & 67.54  & 29.01   & 36.06    & 8.10   & 2.47 \\
13  & 8192  & 149.68  & 46.80   & 15.36   & 36.62    & 6.01   & 1.15 \\
14  & 16384 & 142.17  & 36.14  & 9.13    & 32.79    & 5.94   & 0.14 \\
15  & 32768 & 121.83  & 27.78  & 5.97    & 33.69    & 2.13   & 0.00 \\
\bottomrule
\end{tabular}
\caption{Agent coordination problem: Original search space dimension, $2^n$, and the average reduced search space dimension. Every average was computed across $100$ random graph instances with $n$ nodes and connection probability $p$. Each graph gave rise to a swarm coordination problem, formally given by 3-SAT instance $\Phi$ with $n$ boolean variables and $m$ clauses. The constraints $\Phi_k$ were given by the first $k$ clauses of $\Phi$.
}\label{tab:robots}
\end{table*}

\emph{{Simulation complexity.}} Applying results from Section~\ref{sec:qaoa:red}, we obtain:
\begin{itemize}
    \item In case of 3-SAT, quantum optimization can be simulated in $\calO((\nu + s) 2^{2 n})$ steps on a classical machine, where $\nu$ is the number of stages done by quantum optimization, while $s$ is the number of time steps done during numerical integration of Schrödinger's equation.
    \item If quantum Zeno reduction is applied to a 3-SAT formula $\Phi$, the simulation costs of quantum optimization can be reduced to $\calO((\nu + s) d^2) + \calO(d^2)$, where $d$ is the number of boolean assignments satisfying the constraint $\Phi_k$.
\end{itemize}

The second point follows from Theorem~\ref{thm:qaoa:red} because the Zeno space $Z$ is spanned by the standard basis vectors $e_{d + 1}$ where $\Phi_k(x_d) = 1$ and $x_d \in \{0,1\}^n$ is the binary representation of $0 \leq d \leq 2^n - 1$. Specifically, we work under the assumption that the following common~\cite{farhi2000quantum,farhi2014quantum} Hamiltonian $H_B$ is used. As stated next, its entries can be computed in constant time. 

\begin{lemma}\label{lem:hb:formula}
Let $U_B(\delta) = \exp(-i \delta H_B)$ where $H_B=\sum_{j=1}^{n} X_j$ act on $n$ bits, with $X_j$ flipping bit $j$ and leaving all other bits untouched. Then, the entry of $U_B(\delta)$ at row $k+1$ and column $l+1$ is given by
\[
U_B(\delta)_{k+1,l+1}
=
(\cos \delta)^{\,n-d(k,l)}(-i\sin \delta)^{\,d(k,l)},
\]
where $d(k,l)$ denotes the Hamming distance between the binary representations of integers $k, l \in \{0,\ldots,2^n-1\}$.
\end{lemma}

We note that $H_B$ generalizes the Hamiltonian~(\ref{eq:HB:small}) to an arbitrary number of qubits. As in the special case of $n = 2$, one can show that $U_B(\delta)$ diagonalizes over the basis induced by $\ket{+}$ and $\ket{-}$.

\emph{{Discussion.}} The above discussion implies that one can simulate quantum optimization on classical computers efficiently if $d \ll 2^n$ and if we know the $d$ satisfying boolean assignments of constraint $\Phi_k$. Notice that the latter is an application of \#SAT to $\Phi_k$, i.e., a hard problem itself~\cite{Schaefer1978}. This said, we point out that the current work focuses on the efficient simulation of quantum optimization. Specifically, while currently known algorithms for \#SAT and BQP are both exponential in the number of boolean variables and qubits, respectively, the former usually admit a much milder exponential growth, allowing one to cover sometimes instances with millions of boolean variables, see~\cite{EEN2003,DBLP:conf/cav/InversoT0TP14} and references therein. The current work thus proposes to harness the advances in SAT solving akin to~\cite{DBLP:conf/cav/MeiBL24} to speed-up the simulation of constrained quantum optimization on classical machines.

\section{Conclusion}

We introduced a projection-based reduction technique for simulation of constrained quantum optimization on classical machines. We demonstrated that exponential state-space reductions are possible by applying the approach to a collaborative systems of agents, where feasibility encodes coordination requirements among agents. To this end, we harnessed efficient algorithms for \#SAT which, while being computationally hard themselves, scale beyond currently available approaches in quantum circuit simulation. Future work will use the reported reductions to speed-up constrained quantum optimization.

\appendix

\section*{Appendix with Mathematical Proofs}

\begin{proof}[Theorem~\ref{thm:red:zeno}]
Decompose the initial state as 
\[
\ket{\psi_0} = P_Z \ket{\psi_0} + (I - P_Z)\ket{\psi_0}
\]
and set 
\[
r := (I - P_Z)\ket{\psi_0} \in Z^\perp.
\]
Since $H_Z = P_Z H P_Z$ satisfies $H_Z (I - P_Z) = 0$, we have $H_Z r = 0$, hence $U_Z(t) r = e^{-it H_Z} r = r$, so the component outside $Z$ is invariant under the Zeno evolution. Next define the reduced state $\ket{\widehat{\psi}_t} := S^\dagger \ket{\psi_t}$. Using $P_Z = S S^\dagger$ and $S^\dagger S = I$, we have
\[
P_Z \ket{\psi_t} = S S^\dagger \ket{\psi_t} = S \ket{\widehat{\psi}_t}.
\]
With this, we obtain 
\begin{multline*}
\partial_t \ket{\widehat{\psi}_t}
= S^\dagger \partial_t \ket{\psi_t} 
= -i S^\dagger P_Z H P_Z \ket{\psi_t} \\
= -i S^\dagger H (P_Z \ket{\psi_t}) 
= -i S^\dagger H S \ket{\widehat{\psi}_t} 
= -i \widehat{H}_Z \ket{\widehat{\psi}_t},
\end{multline*}
with initial condition $\ket{\widehat{\psi}_0} = S^\dagger \ket{\psi_0}$. Therefore $\ket{\widehat{\psi}_t} = e^{-it \widehat{H}_Z} \ket{\widehat{\psi}_0}$ and lifting back yields
$
P_Z \ket{\psi_t}
= S \ket{\widehat{\psi}_t}
= S e^{-it \widehat{H}_Z} S^\dagger \ket{\psi_0}
$. 
Finally,
\[
\ket{\psi_t}
= P_Z \ket{\psi_t} + (I - P_Z)\ket{\psi_t}
= S \ket{\widehat{\psi}_t} + r,
\]
since
\[
(I - P_Z)\ket{\psi_t}
= (I - P_Z) U_Z(t) \ket{\psi_0}
= (I - P_Z)\ket{\psi_0}
= r.
\]
To see the second identity, note that  
\[
(I - P_Z)e^{-itH_Z} = (I - P_Z) \Big(I + \sum_{k \ge1 }\frac{(-it)^k}{k!}H_Z^k \Big) = (I - P_Z)
\]
because $(I - P_Z) H_Z^k = 0$ for all $k \geq 1$.
\end{proof}

\begin{proof}[Theorem~\ref{thm:qaoa:red}]
Statement 1. and 2. follow from Theorem~\ref{thm:red:zeno} and the discussion proceeding Theorem~\ref{thm:qaoa:red}. Instead, the statement 3. follows by noting that $H_Z = S^\dagger H S$ for the $2^n \times d$ matrix $S$. Hence, $H_Z$ can be computed by means of $d$ matrix-vector multiplications in $2^n$, given the first complexity bound. Instead, if $S$ consists of standard base vectors only, say $\{e_1,\ldots,e_d\}$, then $H_Z$ arises from $H$ by keeping only the first $d$ columns and first $d$ rows of $H$. In general, with $I = \{i_1,\ldots,i_d\}$ being the indices of standard base vectors appearing in $S$, one can obtain the $d^2$ entries $(H_Z)_{k,l} = H_{i_k,i_l}$ by querying $H$ or by direct calculation of $H_{i_k,i_l}$.
\end{proof}

\begin{proof}[Lemma~\ref{lem:hb:formula}]
Since the operators $X_j$ act on distinct qubits, they commute. Hence
\[
U_B(\delta)=e^{-i\delta\sum_{j=1}^n X_j} =\bigotimes_{j=1}^n e^{-i\delta X}.
\]
Let $k=(k_1,\ldots,k_n)$ and $\ell=(\ell_1,\ldots,\ell_n)$. From tensor-product rules, we obtain
\[
\bra{k}U_B(\delta)\ket{\ell}
=
\bigotimes_{j=1}^n \bra{k_j}e^{-i\delta X}\ket{\ell_j} ,
\]
where $\bra{k} = \ket{k}^\dagger$. Using $X^2=I$, the Taylor expansion of $e^{-i\delta X}$ yields 
\begin{align*}
e^{-i\delta X} & = \sum_{m=0}^{\infty} \frac{(-i\delta X)^m}{m!} \\
& = I - i\delta X + \frac{(-i\delta)^2}{2!}X^2
  + \frac{(-i\delta)^3}{3!}X^3
  + \frac{(-i\delta)^4}{4!}X^4 + \cdots \\
& = I \cos\delta - X i\sin\delta \\
& =
\begin{pmatrix}
\cos\delta & -i\sin\delta\\
-i\sin\delta & \cos\delta
\end{pmatrix}.
\end{align*}
Thus, for $a,b\in\{0,1\}$, we obtain
\[
\bra{a}e^{-i\delta X}\ket{b}
=
\begin{cases}
\cos\delta & a=b,\\
-i\sin\delta & a\neq b .
\end{cases}
\]
Now, among the $n$ positions, $n-d(k,\ell)$ satisfy $k_j=\ell_j$ and $d(k,\ell)$ satisfy $k_j\neq\ell_j$. This yields the claim. 
\end{proof}

\bibliographystyle{splncs04}
\bibliography{biblio}

\end{document}